\documentclass[journal,onecolumn,letterpaper,12pt]{IEEEtran}

\usepackage[T1]{fontenc}
\usepackage[bookmarks]{hyperref}
\hypersetup{colorlinks=true,citecolor=blue,linkcolor=blue,filecolor=blue,urlcolor=blue}
\usepackage{amsmath,amssymb,amsthm}
\usepackage{dsfont}
\usepackage{fullpage}
\usepackage{multirow}
\usepackage{caption}

\usepackage{makecell}
\usepackage{siunitx}

\usepackage{amsthm}
\usepackage{mathtools}

\usepackage{float} 

\usepackage{graphicx}
\usepackage[section]{placeins} 

\usepackage{enumerate}

\usepackage{booktabs}

\allowdisplaybreaks[4]

\numberwithin{equation}{section}

\usepackage{aliascnt}

\newtheorem{theorem}{Theorem}[section]

\newaliascnt{lemma}{theorem}
\newtheorem{lemma}[lemma]{Lemma}
\aliascntresetthe{lemma}

\newaliascnt{proposition}{theorem}
\newtheorem{proposition}[proposition]{Proposition}
\aliascntresetthe{proposition}

\newaliascnt{corollary}{theorem}
\newtheorem{corollary}[corollary]{Corollary}
\aliascntresetthe{corollary}

\newaliascnt{conjecture}{theorem}

\aliascntresetthe{conjecture}

\theoremstyle{definition}

\newaliascnt{definition}{theorem}
\newtheorem{definition}[definition]{Definition}
\aliascntresetthe{definition}

\newtheorem*{remark}{Remark}

\usepackage{cleveref}

\crefname{theorem}{theorem}{theorems}
\Crefname{theorem}{Theorem}{Theorems}

\crefname{lemma}{lemma}{lemmas}
\Crefname{lemma}{Lemma}{Lemmas}

\crefname{proposition}{proposition}{propositions}
\Crefname{proposition}{Proposition}{Propositions}

\crefname{corollary}{corollary}{corollaries}
\Crefname{corollary}{Corollary}{Corollaries}

\crefname{conjecture}{conjecture}{conjectures}
\Crefname{conjecture}{Conjecture}{Conjectures}

\crefname{definition}{definition}{definitions}
\Crefname{definition}{Definition}{Definitions}

\crefformat{enumi}{(#2#1#3)}
\crefmultiformat{enumi}{(#2#1#3)}{ and~(#2#1#3)}{, (#2#1#3)}{ and~(#2#1#3)}
\crefrangeformat{enumi}{(#3#1#4)-(#5#2#6)}
\usepackage{tikz}
\usepackage{pgfplots}
 \pgfplotsset{compat=1.18}
\usepackage{xcolor}
\usepackage{comment}

\definecolor{dgreen}{HTML}{006600}
\definecolor{lgreen}{HTML}{B3FFB3}

\usepackage{caption}
\usepackage{subcaption}
\usepackage{setspace}
\usepackage{bbm}

\newcommand{\cA}{\mathcal{A}}
\newcommand{\cB}{\mathcal{B}}

\newcommand{\cD}{\mathcal{D}}

\newcommand{\cH}{\mathcal{H}}

\newcommand{\cL}{\mathcal{L}}
\newcommand{\cM}{\mathcal{M}}
\newcommand{\cN}{\mathcal{N}}

\newcommand{\cQ}{\mathcal{Q}}

\newcommand{\cT}{\mathcal{T}}

\newcommand{\cY}{\mathcal{Y}}

\newcommand{\bC}{\mathbb{C}}

\newcommand{\eps}{\varepsilon}

\DeclareMathOperator{\tr}{tr}

\DeclareMathOperator{\id}{id}

\DeclareMathOperator{\Ad}{Ad}

\newcommand{\epsstar}{\varepsilon_*}

\usepackage{ytableau}
\usepackage{physics}
\usepackage{bm}
\newcommand{\be}{\begin{equation}}
\newcommand{\ee}{\end{equation}}

\newcommand{\proj}[1]{|#1\rangle\langle #1|}

\usepackage[backend=biber,sorting=none,style=numeric-comp,giveninits=true,doi=false,isbn=false,url=false,maxbibnames=20,maxcitenames=2]{biblatex}
\renewbibmacro{in:}{}
\newbibmacro{string+doi}[1]{\iffieldundef{doi}{#1}{\href{https://dx.doi.org/\thefield{doi}}{#1}}}
\DeclareFieldFormat{title}{\usebibmacro{string+doi}{\mkbibemph{#1}}}
\DeclareFieldFormat[article]{title}{\usebibmacro{string+doi}{\mkbibquote{#1}}}
\DeclareFieldFormat[incollection]{title}{\usebibmacro{string+doi}{\mkbibquote{#1}}}                   
\DeclareFieldFormat[inproceedings]{title}{\usebibmacro{string+doi}{\mkbibquote{#1}}}

\newcommand{\C}{\mathbb{C}}

\newcommand{\bea}{\begin{eqnarray}}
\newcommand{\eea}{\end{eqnarray}}
\newcommand{\beas}{\begin{eqnarray*}}
\newcommand{\eeas}{\end{eqnarray*}}

\usepackage{bm}
\numberwithin{equation}{section}

\begin{document}
\title{High quantum local differential privacy \\ breaks entanglement}

 \author{
\IEEEauthorblockN{Sujeet Bhalerao\IEEEauthorrefmark{1},
Theshani Nuradha\IEEEauthorrefmark{1}\IEEEauthorrefmark{2},
Felix Leditzky\IEEEauthorrefmark{1}\IEEEauthorrefmark{2}
}

\IEEEauthorblockA{\IEEEauthorrefmark{1}
\small Department of Mathematics, University of Illinois Urbana-Champaign, Urbana, IL 61801, USA
}

\IEEEauthorblockA{\IEEEauthorrefmark{2}
\small Illinois Quantum Information Science and Technology (IQUIST) Center,\\
University of Illinois Urbana-Champaign, Urbana, IL 61801, USA
}
}

\maketitle
\date{\today}
\begin{abstract}
Differential privacy provides a mathematical framework for making use of sensitive data while providing a privacy guarantee for what can be learned about the input data. In quantum information processing, the interaction of privacy constraints with quantum resources such as entanglement remains a question of interest. Given that the utility of many protocols, and often the presence of a quantum advantage, relies on quantum resources such as entanglement, it is crucial to understand when a privacy requirement for a quantum channel is compatible with the channel's ability to preserve entanglement. We study this question for quantum local differential privacy (QLDP). Our main result shows that every $\varepsilon$-QLDP channel with a $d$-dimensional input is entanglement-breaking whenever $\varepsilon\leq\log\frac{d}{d-1}$. We also prove an approximate version for $(\varepsilon,\delta)$-QLDP, where channels in the same high-privacy regime are close in diamond norm to an entanglement-breaking channel. We further prove a composition result for a collection of private quantum channels having entangled inputs and global measurements in the high-privacy regime. Finally, we apply our results to private quantum learning theory. We prove that any learning protocol using arbitrary quantum memory on copies of the output of an entanglement-breaking channel can be simulated by a protocol that measures the corresponding unprocessed input copies one at a time while storing only classical information. Combining this result with our high-privacy entanglement-breaking theorem, we show that under sufficiently private local noise, a learning protocol with quantum memory for purity testing and bipartite product testing is subject to the sample complexity lower bounds for protocols with single-copy measurements on the noiseless tasks. We also obtain stronger sample complexity lower bounds when a single highly private channel acts on the entire multipartite input.

\end{abstract}

\tableofcontents

\section{Introduction}

As quantum processors and quantum communication networks are implemented at increasing scales, the protection of sensitive quantum data has become an emerging problem of interest. A user may send quantum information encoded as a state to a remote processor, or several users may contribute quantum systems to a distributed protocol. In either case, one would like to use the data while limiting what can be learned about any sensitive information encoded in the underlying input state. Differential privacy provides a mathematical framework to formalize such a requirement \cite{DMNS06,DR14}.

Quantum differential privacy replaces the classical randomized mechanism by a quantum channel. Early formulations studied privacy in quantum computation and its relation to gentle measurement \cite{QDP_computation17,aaronson2019gentle}. Information-theoretic formulations were developed in \cite{hirche2023quantum}, while the local model has been studied further in \cite{angrisani2023unifying,angrisani_localModel25}. More flexible notions of quantum privacy have also been considered, including quantum Pufferfish privacy \cite{nuradha_QPP,measuredHS25}. Recent works have also studied the composition of multiple quantum differentially private mechanisms \cite{QDP_computation17,hirche2023quantum,nuradha_QPP,guan2024optimal,alabi2026quantum}. Optimal mechanisms in the local model have been investigated in \cite{yoshida2026optimal,guan2024optimal, yoshida-hayashi}, and an information-ordering approach was developed in \cite{dasgupta2025quantum}.

In this paper, we consider quantum local differential privacy (QLDP). For an $\varepsilon$-QLDP channel (defined in \Cref{def:QLDP} below), the probability of any measurement outcome on one input state can differ from the corresponding probability on another input state by at most a factor of $e^\varepsilon$. When the privacy parameter $\varepsilon$ is small, the outputs of the channel must therefore be difficult to distinguish. Previous work has quantified this loss of distinguishability using contraction bounds and private hypothesis testing \cite{nuradha2024contraction,Christoph2024sample,Farhad_QP_HT,namPUT_HT_2025quantum,,nuradha2025nonLinearSDPI}. Recently, we focused on the tradeoff between privacy and utility instead, asking how well a channel can preserve a chosen utility while satisfying the same privacy constraint \cite{nuradha2026privacyutilitytradeoffsQI} (see also~\cite[Section~6.3]{gallage2025theory}). Entangled inputs under local mechanisms and local measurements were studied in \cite{wang2026entanglement}.

Choosing the privacy parameter $\varepsilon$ is an important part of using differential privacy. 
In the classical setting, this has been studied in the context of how to interpret $\varepsilon$ and how its numerical value should enter practical decisions about privacy \cite{ExplainEpsilon23,nanayakkara2024consider}. A key distinguishing feature in quantum information processing that does not exist in classical information processing is the presence of quantum entanglement. It is therefore natural to ask how the choice of the privacy parameter that gives strong protection of the input interacts with the presence of quantum resources such as entanglement. Furthermore, quantum information processing tasks often make use of entanglement or another quantum resource, such as a quantum memory, to prove a quantum advantage in some learning-theoretic task. 
This leads to the question how such protocols are affected in the presence of a privacy constraint in the form of a QLDP quantum channel. We study to what extent channels satisfying QLDP can preserve or break entanglement and how this behaviour depends on the choice of the privacy parameter.

\subsection{Main results}
\noindent\textbf{High privacy breaks entanglement.} Our main result gives a threshold for the privacy parameter $\varepsilon$ below which every $\varepsilon$-QLDP channel is entanglement-breaking. We prove that any $\varepsilon$-QLDP quantum channel $\cN\colon \cL(A)\to \cL(B)$ with $d=\dim A$ is entanglement-breaking whenever 
$$
\varepsilon\leq\log\frac{d}{d-1}.
$$
For a qubit input, this gives the threshold $\varepsilon\leq\log 2$ for the privacy parameter. The constant is optimal in the following sense: for every $\eta>0$, there is a channel which is not entanglement-breaking and whose optimal QLDP parameter is less than
$$\log\frac{d}{d-1}+\eta. $$
This gives one possible interpretation of the privacy parameter in the quantum differential privacy setting: if a protocol needs to preserve entanglement with a reference system, then choosing $\varepsilon$ below this threshold rules out that possibility.

The same argument also applies to the approximate privacy setting. If $\cN$ satisfies $(\varepsilon,\delta)$-QLDP and $\varepsilon\leq\log(d/(d-1))$, then there is an entanglement-breaking channel $\cM$ such that
$$
\frac12\|\cN-\cM\|_\diamond
\leq
(d-1)\delta.
$$

The threshold also depends on the dimension of the reference system with which entanglement is considered, in the following sense. Suppose that one only asks whether the channel must break entanglement with a reference system of dimension $r$. Then every $\varepsilon$-QLDP channel is $r$-entanglement-breaking when
$$\varepsilon\leq\log\frac{r}{r-1}.$$
In particular, $\varepsilon\leq\log 2$ also implies that the channel breaks entanglement with a qubit reference system even if the input dimension is larger than two. 

\medskip
\noindent\textbf{Composition with entangled inputs and global measurements with high privacy.}
We study tensor products of highly private channels when the joint input
may be entangled across the local systems and the output may be
measured globally. Let
$\cN_i\colon \cL(A_i)\to\cL(B_i)$ be
$\varepsilon_i$-QLDP, let $d_i=\dim A_i$, and suppose
$$\varepsilon_i<\log\frac{d_i}{d_i-1}.
$$
Define
$$\beta_i = \frac{
d_i\bigl(d_i-(d_i-1)e^{\varepsilon_i}\bigr)
}{1+(d_i-1)e^{\varepsilon_i}
}, \qquad
\gamma_i = \frac{
d_i\bigl(d_i-(d_i-1)e^{-\varepsilon_i}\bigr)
}{1+(d_i-1)e^{-\varepsilon_i}}.$$
We prove that $\bigotimes_i\cN_i$ is
$\varepsilon_{\mathrm{comp}}$-QLDP with respect to arbitrary pairs
of joint input states, where
$$
\varepsilon_{\mathrm{comp}}
=
\sum_i\log\frac{\gamma_i}{\beta_i}.
$$ Recent work established composition guarantees for tensor-product
channels acting on product neighboring inputs
\cite{alabi2026quantum}. In contrast, our composition theorem
allows arbitrary entangled pairs of joint input states and arbitrary
global output measurements, under the additional assumption that
each local channel lies in the high-privacy regime. For $n$ identical qubit-input channels with privacy parameter $\varepsilon<\log 2$, the composed privacy parameter is $$\varepsilon_{\mathrm{comp}}=
n\log\left(\frac{2e^\varepsilon-1}{2-e^\varepsilon}\right),$$
which is linear in $n$ and grows as $3n\varepsilon+O(n\varepsilon^3)$ for small $\varepsilon$. 

\medskip
\noindent\textbf{Applications to private quantum learning theory.}
We first prove a general simulation result for entanglement-breaking noise that may be of independent interest: any binary-output protocol that uses arbitrary quantum memory on copies of the output of an entanglement-breaking channel can be simulated by a protocol that measures the corresponding unprocessed input copies one at a time storing only classical information. Since every $\varepsilon_i$-QLDP channel on a $d_i$-dimensional input is entanglement-breaking when $\varepsilon_i\leq \log\frac{d_i}{d_i-1},$ the known single-copy lower bounds for the corresponding noiseless learning tasks also apply to quantum memory learners acting on the privatized states. For the purity-testing task of~\cite{CCHL}, this gives $T=\Omega(2^{n/2})$ when each of the $n$ input qubits is subjected to local highly private noise. For bipartite product testing on $n$ systems of local dimension $q$, the lower bound of~\cite{beckey2025product} similarly gives $T=\Omega(q^{n/4})$. We also consider the setting in which a single highly private channel acts on the entire input. Directly from the privacy condition, we obtain the stronger lower bounds $T=\Omega(4^n)$ for purity testing and $T=\Omega(q^{2n})$ for bipartite product testing. While prior work \cite{Cotler2026noisy} showed that local depolarizing noise can eliminate the two-copy advantage for purity testing, our result applies to all entanglement-breaking channels and to any protocol with quantum memory.

Additionally, in Appendix \ref{app: geometric-privacy} we connect QLDP channels to a result about the separable ball around the maximally mixed state due to Gurvits and Barnum \cite{gurvits2002largest}.

\subsection{Structure of the paper}
In \Cref{sec: bckdgrnd} we fix notation and recall the definition of quantum local differential privacy. In \Cref{sec:main-result} we prove
the entanglement-breaking privacy threshold and its approximate version, and also
discuss its optimality and bounded reference systems. In \Cref{sec: 4-composition} we derive the
composition bounds with entangled inputs and global measurements. 
In \Cref{sec: 5-learning}  we apply our entanglement-breaking result to purity
testing and bipartite product testing with quantum memory. \Cref{sec: 6-conclusion} concludes with a discussion of our results. In Appendix \ref{app: geometric-privacy} we provide a geometric interpretation of quantum local differential privacy, in particular to the ball of separable states around a separable Choi state.

\section{Background}\label{sec: bckdgrnd}
\subsection{Notation}
All Hilbert spaces we consider are finite-dimensional over $\bC$. For a system $A$, we denote by $\cD(A)$ the set of quantum states on $A$, which are positive semidefinite operators with unit trace. A state is called pure if it has rank one; otherwise we say it is mixed. Let $\cL(A)$ denote the linear operators on $A$. A quantum channel $\cN\colon\cL(A)\to\cL(B)$ is a completely positive trace-preserving linear map from $\cL(A)$ to $\cL(B)$. We write $\pi_A=I_A/d_A$ for the maximally mixed state on $A$, where $d_A=\dim A$ and $I_A$ is the identity operator on $A$.
We also denote by $\id_A$ the identity map on $\cL(A)$.
The trace norm of an operator $X$ is defined as $\left\|X\right\|_1 \coloneqq \Tr[\sqrt{X^\dagger X} ]$. 
The diamond norm of a linear map $\Phi\colon \cL(A)\to\cL(B)$ is defined as
\begin{align}
    \|\Phi\|_\diamond = \sup_{X_{RA}} \frac{\|(\id_R\otimes \Phi)(X_{RA})\|_1}{\|X_{RA}\|_1},
\end{align}
where the supremum is taken over all reference systems $R$ with $\dim R\leq \dim A$ and linear operators $X_{RA}$.
For operators $X$ and $Y$, we write $X \geq Y$ to indicate that $X-Y$ is a positive semi-definite (PSD) operator, while $X > Y$ indicates that $X-Y$ is positive definite. A positive operator-valued measure (POVM) is a collection of PSD operators $\{M_y\}_{y \in \cY}$ satisfying $\sum_{y \in \cY} M_y= I_{R}$, where $\cY$ is some finite set.

\subsection{Quantum local differential privacy}
 We recall the definition of quantum local differential privacy (QLDP) below~\cite{hirche2023quantum,nuradha_QPP}. 
\begin{definition}\label{def:QLDP}
    Fix $\varepsilon \geq 0$ and $\delta \in [0,1]$. 
    Let $\cA$ be a quantum channel. We say $\cA$ is \emph{$(\varepsilon, \delta)$-quantum local differentially private}, or $(\varepsilon,\delta)$-QLDP, if 
\begin{equation} \Tr\left[M \cA(\rho)\right] \leq e^\varepsilon \Tr\left[M \cA(\sigma)\right] + \delta
\label{eq:QLDP-def}
\end{equation}
for all pairs of states $\rho, \sigma$ and measurement operators $M$ satisfying $0 \leq M \leq I.$
We say that $\cA$ satisfies $\varepsilon$-QLDP if it satisfies $(\varepsilon,0)$-QLDP.
\end{definition}

It is useful to define the optimal privacy parameter of a channel:
$$
\epsstar(\cN)= \inf\{\varepsilon\ge0:\cN \text{ is } \varepsilon\text{-QLDP}\}.
$$
Thus $\cN$ is $\varepsilon$-QLDP exactly when $\epsstar(\cN)\le\varepsilon$. A channel $\cN\colon \cL(A)\to \cL(B)$ is \emph{entanglement-breaking} if, for any reference system $R$, the channel $\id_R \otimes \cN$ maps every state on $R\otimes A$ to a separable state across $R:B$. Every entanglement-breaking channel has a measure-and-prepare representation, that is, it can be written as \cite{horodecki2003entanglementbreaking}
$$\mathcal N(\rho)
=
\sum_{y\in\mathcal Y}\operatorname{Tr}[E_y\rho]\,\sigma_y,$$ for a POVM $\{E_y\}_{y \in Y}.$

The following elementary consequence of the definition of QLDP will be useful for computing the exact privacy parameters of several examples.

\begin{lemma}
\label{lem:spectral-epsstar}
Let $\cN\colon \cL(A)\to\cL(B)$ be a channel. Suppose there are constants $0<m\le u$ such that $mI_B\le\cN(\rho)\le uI_B$ for every state $\rho$. Suppose also that there are states $\rho_0,\sigma_0$ and a measurement $P$ such that $\Tr[P\cN(\rho_0)]=u$ and $\Tr[P\cN(\sigma_0)]=m$. Then $\epsstar(\cN)=\log(u/m)$.
\end{lemma}

\begin{proof}
By the assumption on $\cN$, we have for every pair of states $\rho,\sigma$ and every measurement
$0\leq M\leq I_B$ that
$$\Tr[M\cN(\rho)]\leq u\Tr[M] \qquad \text{and} \qquad
m\Tr[M]\leq\Tr[M\cN(\sigma)].
$$
Thus,
$$\Tr[M\cN(\rho)]\leq\frac{u}{m}\Tr[M\cN(\sigma)],
$$
so $\cN$ is $\log(u/m)$-QLDP. Conversely, the states
$\rho_0,\sigma_0$ and the measurement $P$ give the likelihood ratio
$$
\frac{\Tr[P\cN(\rho_0)]}{\Tr[P\cN(\sigma_0)]}=\frac{u}{m}.
$$
No smaller privacy parameter is therefore possible.
\end{proof}

\section{High local privacy breaks entanglement}\label{sec:main-result}

In this section, we prove our main results on quantum channels in the high privacy regime of QLDP being necessarily entanglement-breaking. First, we briefly outline the proof idea. We use the fact that, if $d=\dim A$, $\psi=|\psi\rangle\langle\psi|$, and $\mathbf{d}\psi$ is the Haar measure on pure states, then every operator $X\in \cL(A)$ satisfies
$$X = d\int \Tr[\psi X]\bigl((d+1)\psi-I_A\bigr)\,\mathbf{d}\psi .$$
For a state $\rho$, the coefficient $d\,\Tr[\psi\rho]\,\mathbf{d}\psi$ is a probability measure. The proof of the main theorem shows that, under the high-privacy condition, applying the channel to each operator $(d+1)\psi-I_A$ gives a density operator, and thus we obtain a measure-and-prepare form for the channel. This is equivalent to the channel being entanglement-breaking by \cite{horodecki2003entanglementbreaking}.

\begin{theorem}\label{thm: main}
Let $\cN\colon \cL(\cH_A)\to \cL(\cH_B)$ be a quantum channel. Suppose that $\cN$ is $\varepsilon$-QLDP. If $$\varepsilon\le \log (d_A/(d_A-1)),$$ then $\cN$ is entanglement-breaking. Consequently, if $\cN$ is not entanglement-breaking, then $\varepsilon > \log (d_A/(d_A-1))$.
\end{theorem}

\begin{proof}
Let $d\coloneqq d_A$.
    First recall that 
    \begin{equation}
      d\int_{|\psi \rangle \in \mathbb{S}^{d}} \langle \psi| \rho |\psi\rangle \   | \psi\rangle\langle \psi| \ \mathbf{d}\psi = \frac{I_{A} +\rho}{d+1},
      \label{eq:integral}
    \end{equation}
    which can be seen by rewriting the integrand as $\langle \psi| \rho |\psi\rangle \   | \psi\rangle\langle \psi| = \Tr_2\left[(I_A\otimes \rho)|\psi\rangle\langle\psi|^{\otimes 2}\right]$ and using $\int_{|\psi \rangle \in \mathbb{S}^{d}} |\psi\rangle\langle\psi|^{\otimes 2} \ \mathbf{d}\psi = \frac{1}{d(d+1)}(I_{AB} + F_{AB})$, where $F_{AB}$ denotes the swap operator.
    
With $\psi \coloneqq | \psi\rangle\langle \psi|$, the equation \eqref{eq:integral} can be rewritten as
\begin{align}\label{eq:single-expansion}
    \rho ={}&  \int_{\psi} d\ \Tr[\psi \rho] \left( (d+1) \psi -I_{A} \right) \ \mathbf{d}\psi \\
    R_{\psi} \coloneqq{}& (d+1) \psi -I_{A} \\
    \cN(\rho) ={}&  \int_{\psi} d\ \Tr[\psi \rho] \ \cN(R_{\psi}) \ \mathbf{d}\psi.
\end{align}
With these, let us define $\tau_{\psi} \coloneqq \cN(R_{\psi}) $. Since $\cN$ is trace-preserving and $\Tr[R_{\psi}]= d+1-d=1$, we also have $\Tr[\tau_{\psi}]=1$.

Moreover, for $0\leq M \leq I_B$, we have that $0 \leq A \coloneqq \cN^\dag(M)  \leq I_{A}$ since $\cN^\dag$ is a unital completely positive map. We also have that $\cN$ satisfies $\varepsilon$-QLDP leading to 
$\Tr[ \cN^\dag(M) \rho] \leq e^\varepsilon \Tr[ \cN^\dag(M) \sigma]$ for all $\rho,\sigma$ states. By picking the states generated by the eigenvectors of the operator $A$, that is, $| \lambda_i\rangle\!\langle \lambda_i|$ with the corresponding eigenvalue $\lambda_i \geq 0$ for $i\in\{1,\ldots, d\}$, we get 
$\lambda_i \leq e^\varepsilon \lambda_j$.
Consider then 
\begin{align}
    \Tr[ M \cN(R_\psi)] &= \Tr[ \cN^\dag(M) R_{\psi}]\\
    & =\Tr[ AR_{\psi}] \\
    &= (d+1) \Tr[A \psi] -\Tr[A] \\
    &\geq (d+1) \lambda_{\min} -\left( \lambda_{\min}(A) +(d-1) e^\varepsilon \lambda_{\min}(A)\right) \\
    &= (d- (d-1) e^{\varepsilon}) \lambda_{\min}(A).
\end{align}
Thus, whenever $e^{\varepsilon} \leq d/(d-1) $, for all $M$, we have that 
$$ \Tr[M \tau_{\psi}] = \Tr[M  \cN(R_\psi)] \geq 0$$ along with $A\geq 0$. This results in $ \tau_{\psi}=\cN(R_{\psi}) \geq 0$. Together with $\Tr[\tau_\psi]=1$, this means that $\tau_{\psi}$ is a quantum state. 
Then, for all input states $\rho$,
\begin{equation}
    \cN(\rho)=  \int_{\psi} d\ \Tr[\psi \rho] \ \tau_{\psi} \ \mathbf{d}\psi,
\end{equation}
which represents the action of an entanglement-breaking channel.
\end{proof}

\begin{proposition}
    Let $\cN\colon \cL(\cH_A)\to \cL(\cH_B)$ be a quantum channel.
     Suppose that $\cN$ is $(\varepsilon,\delta)$-QLDP. If $$\varepsilon\le \log (d_{A}/(d_{A}-1)),$$ then there exists an entanglement-breaking channel $\cM$ such that $  \frac{1}{2}  \| \cN-\cM\|_\diamond \leq \delta (d_{A}-1). $
\end{proposition}
\begin{proof} Let $d\coloneqq d_A$.
Defining $A=\cN^\dagger(M)$ and denoting its eigenvalues by $\lambda_i$ as in the proof of \Cref{thm: main}, we have $\lambda_i \leq e^\varepsilon \lambda_j + \delta$, since the channel $\cN$ satisfies $(\varepsilon,\delta)$-QLDP and we can choose $\rho$ and $\sigma$ to be the pure states corresponding to the respective eigenvectors.
    
    Then, consider 
    \begin{align}
    \Tr[ M \tau_{\psi}] &= \Tr[ \cN^\dag(M) R_{\psi}]\\
    &= \Tr[ AR_{\psi}] \\
    &= (d+1) \Tr[A \psi] -\Tr[A] \\
    &\geq (d+1) \lambda_{\min} -( \lambda_{\min} +(d-1) (e^\varepsilon \lambda_{\min} + \delta)) \\
    &= (d- (d-1) e^{\varepsilon}) \lambda_{\min} - (d-1) \delta \\
    & \geq - (d-1) \delta,
\end{align}
where the last inequality follows from $\varepsilon\le \log (d/(d-1))$.
Note that 
$\Tr[(\tau_{\psi})_{-}] \leq (d-1) \delta$ since we can choose $M$ to be the projection onto the negative eigenspace of $\tau_{\psi}$.

Let $\eta_\psi\coloneqq \Tr[(\tau_\psi)_-]$. Since $\Tr[\tau_\psi]=1$, also $\Tr[(\tau_\psi)_+]=1+\eta_\psi$. Define 
$$\sigma_\psi\coloneqq \frac{1}{\Tr[(\tau_\psi)_+]}(\tau_\psi)_+=\frac{1}{(1+\eta_\psi)}(\tau_\psi)_+.$$ 
Then $\sigma_\psi$ is a state. Moreover, since $(\tau_\psi)_+$ and $(\tau_\psi)_-$ have orthogonal supports, $\|\tau_\psi-\sigma_\psi\|_1=2\eta_\psi\le 2(d-1)\delta$.

Now define the channel $\cM$ by $\cM(X)=\int d\,\Tr[\psi X]\sigma_\psi\,\mathbf{d}\psi$, where $\mathbf{d}\psi$ is normalized Haar measure on pure states and $d$ is the input dimension. This is a measure-and-prepare channel, so it is entanglement-breaking.

It remains to bound $\|\cN-\cM\|_\diamond$. Let $R$ be an arbitrary reference system, and let $\omega_{RA}$ be a state. We have $(\id_R\otimes(\cN-\cM))(\omega_{RA})=\int d\,\Tr_A[(I_R\otimes\psi)\omega_{RA}]\otimes(\tau_\psi-\sigma_\psi)\,\mathbf{d}\psi$. By the triangle inequality and multiplicativity of the trace norm,
$$\|(\id_R\otimes(\cN-\cM))(\omega_{RA})\|_1\le \int d\,\|\Tr_A[(I_R\otimes\psi)\omega_{RA}]\|_1\|\tau_\psi-\sigma_\psi\|_1\,\mathbf{d}\psi.$$

For each $\psi$, the operator $\Tr_A[(I_R\otimes\psi)\omega_{RA}]$ is positive, so its trace norm equals its trace. Therefore
$$\|(\id_R\otimes(\cN-\cM))(\omega_{RA})\|_1\le 2(d-1)\delta\int d\,\Tr[\Tr_A[(I_R\otimes\psi)\omega_{RA}]]\,\mathbf{d}\psi.$$
The integral is equal to $$\Tr[\left(I_R\otimes d\int \psi\,\mathbf{d}\psi\right)\omega_{RA}]=\Tr[\omega_{RA}]=1,$$ since $d\int\psi\,\mathbf{d}\psi=I_A$. Thus $$\|(\id_R\otimes(\cN-\cM))(\omega_{RA})\|_1\le 2(d-1)\delta.$$

Taking the supremum over all reference systems $R$ and all states $\omega_{RA}$ gives $\|\cN-\cM\|_\diamond\le 2(d-1)\delta$. Equivalently, $\frac12\|\cN-\cM\|_\diamond\le(d-1)\delta$, as claimed.
\end{proof}

Next, we show that the threshold $\log(d/(d-1))$ for the privacy parameter stated in the theorems above is optimal in the following sense: for any privacy parameter $\eps$ strictly larger than this value, one can always find a channel that is $\eps$-QLDP and \emph{not} entanglement-breaking.

\begin{theorem}
\label{thm:sharp}
For every $d\ge2$, the constant $\log(d/(d-1))$ in \Cref{thm: main} is optimal: for every $\eta>0$, there is a channel that is not entanglement-breaking and is $\varepsilon$-QLDP for some $\varepsilon<\log(d/(d-1))+\eta$.
\end{theorem}

\begin{proof}
Consider the transpose-depolarizing channel $\cT_t(\rho)=t\rho^T+(1-t)\Tr(\rho)I/d$. Its normalized Choi state is $$J(\cT_t)=tF/d+(1-t)I_{d^2}/d^2,$$ where $F$ is the swap operator. It is known that $\cT_t$ is completely positive exactly when $-1/(d-1)\le t\le1/(d+1)$.

The state $J(\cT_t)$ is a Werner state, which is known to be separable precisely for $t\ge t_0$, where $t_0=-1/(d^2-1)$ \cite{werner1989EPR}.

Now we compute the exact privacy parameter for $t<0$, restricting to $t>-1/(d-1)$ so that all outputs are full rank. If $r_1,\ldots,r_d$ are the eigenvalues of $\rho$, then $0\le r_i\le1$ and the eigenvalues of $\cT_t(\rho)$ are $(1-t)/d+tr_i$. Since $t<0$, these eigenvalues lie between $m_t=(1+(d-1)t)/d$ and $u_t=(1-t)/d$. Thus $m_tI\le\cT_t(\rho)\le u_tI$ for every state $\rho$.

The ratio $u_t/m_t$ is actually attained: Choose two orthonormal basis vectors $x,y$, take $\rho_0=\proj{y}$, $\sigma_0=\proj{x}$, and take the measurement $P=\proj{x}$. Then $\Tr[P\cT_t(\rho_0)]=u_t$, because $x$ is orthogonal to the support of $\rho_0^T$, while $\Tr[P\cT_t(\sigma_0)]=m_t$, because $x$ lies in the support of $\sigma_0^T$. \Cref{lem:spectral-epsstar} therefore gives
$$
\epsstar(\cT_t)=\log\frac{u_t}{m_t}=\log\frac{1-t}{1+(d-1)t}.
$$
At the entanglement-breaking boundary $t_0=-1/(d^2-1)$, this ratio is $d/(d-1)$. Indeed, $u_{t_0}=d/(d^2-1)$ and $m_{t_0}=(d-1)/(d^2-1)$.

Finally, let $\eta>0$. Since $$t\mapsto \log((1-t)/(1+(d-1)t))$$ is continuous at $t_0$, we may choose $t<t_0$, still inside the complete-positivity interval, such that $$\epsstar(\cT_t)<\log(d/(d-1))+\eta.$$ Because $t<t_0$, the Choi state is a Werner state with $\Tr[FJ(\cT_t)]<0$, hence is entangled. Thus $\cT_t$ is not entanglement-breaking. This proves optimality.
\end{proof}

\begin{remark}
    Our proof of \Cref{thm: main} works in every dimension $d$, but there is a useful simplification for a channel with qubit input, which we now describe. 
    Let $u_1,\ldots,u_4\in \mathbb R^3$ be unit vectors satisfying
$$\sum_{i=1}^4 u_i=0,
\qquad
u_i\cdot u_j=-\frac13
\quad
(i\neq j).$$
Let $\sigma=(\sigma_x,\sigma_y,\sigma_z)$ be the Pauli matrices and define
$$ E_i=\frac14(I_2+u_i\cdot \sigma),
\qquad
R_i=\frac12(I_2+3u_i\cdot \sigma).$$
Then $\{E_i\}_{i=1}^4$ is a POVM and every qubit state $\rho$ satisfies
$$\rho=\sum_{i=1}^4 \Tr[E_i\rho]R_i.$$
If $\cN\colon \cL(\mathbb C^2)\to \cL(B)$ is $\varepsilon$-QLDP with
$\varepsilon\le \log 2,$
then $\cN(R_i)\ge 0$ for every $i$. Indeed, for $0\le M\le I_B$, write
$$
\cN^\dagger(M)=aI_2+b\cdot \sigma.
$$
The eigenvalue ratio forced by QLDP gives
$$a+|b|\le e^\varepsilon(a-|b|)\le 2(a-|b|),$$
hence $a\ge 3|b|$. Therefore
$$\Tr[M\cN(R_i)] =
\Tr[\cN^\dagger(M)R_i] =
a+3b\cdot u_i \ge a-3|b|
\ge 0.
$$
Thus $\cN(R_i)$ is a state, and
$$
\cN(\rho)= \sum_{i=1}^4 \Tr[E_i\rho]\cN(R_i)
$$
is a measure-and-prepare representation of $\cN$.
\end{remark}

\subsection{Privacy thresholds for entanglement with bounded reference systems}
\label{sec:bounded-reference-hierarchy}

The entanglement-breaking threshold in \Cref{thm: main} can be distinguished from the thresholds obtained when the reference system has bounded dimension. We use the notion of bounded-reference entanglement-breaking introduced in~\cite{christandl2018composed}, which we recall below.

\begin{definition}
\label{def:r-eb}
Let $\cN\colon \cL(A)\to\cL(B)$ be a channel and let $1\le r\le d_A$. We say that $\cN$ is \emph{$r$-entanglement-breaking}, or $r$-EB, if $(\id_r\otimes\cN)(\omega)$ is separable
for every state $\omega\in\cD(\mathbb C^r\otimes A)$.
\end{definition}

For an isometry $V\colon\mathbb C^r\to A$, write $\Ad_V(X)=VXV^\dagger$. We use the following characterization of $r$-entanglement-breaking channels from \cite{chen2019entanglementbreaking}.  Its rank-$r$ projection form appears in~\cite[Lemma~6.1]{chen2019entanglementbreaking}, and we use the equivalent form using an isometry. These are equivalent since every rank-$r$ projection $P$ is of the form $P = VV^\dagger$ for an isometry $V\colon \C^r \to A.$ 

\begin{lemma}[{\cite[Lemma~6.1]{chen2019entanglementbreaking}}]
\label{lem:r-eb-restrictions}
A channel $\cN\colon \cL(A)\to\cL(B)$ is $r$-EB if and only if $\cN\circ\Ad_V$ is entanglement-breaking for every isometry $V\colon \mathbb C^r\to A$.
\end{lemma}

\begin{theorem}
\label{thm:r-eb-hierarchy}
Let $\cN\colon \cL(A)\to\cL(B)$ be an $\varepsilon$-QLDP channel and let $2\le r\le d_A$. If
\begin{equation}
 \varepsilon\le\log\frac{r}{r-1},
\end{equation}
then $\cN$ is $r$-entanglement-breaking. 
\end{theorem}

\begin{proof}
Let $V\colon \mathbb C^r\to A$ be an isometry. The restricted channel $\cN_V=\cN\circ\Ad_V$ remains $\varepsilon$-QLDP, because its admissible input states form a subset of those for $\cN$. Its input dimension is $r$, so \Cref{thm: main} shows that $\cN_V$ is entanglement-breaking whenever $\varepsilon\le\log(r/(r-1))$. Since this holds for every $V$, by \Cref{lem:r-eb-restrictions} $\cN$ is $r$-EB.
\end{proof}

In particular, this implies that every $\eps$-QLDP channel with $\eps \leq \log(2)$ breaks entanglement with a qubit reference system, regardless of the dimension of the input to the channel.

\subsection{Privacy of entanglement-breaking channels}

In this subsection, we focus on the potential privacy parameters of entanglement-breaking channels. In fact, we show that for any $\varepsilon$ there is an entanglement-breaking channel that is $\eps$-QLDP.

\begin{proposition}
 For every $\varepsilon_0\in[0,\infty]$, there is an entanglement-breaking channel $\cN_{\varepsilon_0}\colon\cL(A)\to \cL(B)$ with $\varepsilon^*(\cN_{\varepsilon_0})=\varepsilon_0$.
\end{proposition}
\begin{proof}
First assume $0\leq \varepsilon_0<\infty$. Let $q=e^{\varepsilon_0}/(1+e^{\varepsilon_0})$. Then $1/2\leq q<1$ and $q/(1-q)=e^{\varepsilon_0}$. Define $$\tau_0=q|0\rangle\langle 0|+(1-q)|1\rangle\langle 1|$$ and $$\tau_1=(1-q)|0\rangle\langle 0|+q|1\rangle\langle 1|.$$

Let $P$ be some non-trivial projection with $1\leq \tr P\leq d_A-1$ on $A$ and define the channel $$\cN_{\varepsilon_0}(\rho)=\Tr[P\rho]\tau_0+\Tr[(I-P)\rho]\tau_1.$$ 
This is a measure-and-prepare channel, so it is entanglement-breaking. We want to show that $\cN_{\varepsilon_0}$ has $\varepsilon^*(\cN_{\varepsilon_0})= \varepsilon_0$. We will do so in two steps: first we show that $\varepsilon^*(\cN_{\varepsilon_0}) \leq \varepsilon_0$ and then we show the equality.

Fix states $\rho,\sigma\in D(A)$ and a measurement operator $0\leq M\leq I_B$. Let $p=\Tr[P\rho]$ and $s=\Tr[P\sigma]$. Also let $\alpha=\langle 0|M|0\rangle$ and $\beta=\langle 1|M|1\rangle$. Since $0\leq M\leq I_B$, we have $\alpha,\beta\geq 0$.

The outputs of $\cN_{\varepsilon_0}$ are diagonal in the computational basis. Thus $\Tr[M\cN_{\varepsilon_0}(\rho)]=x\alpha+(1-x)\beta$ with $x=pq+(1-p)(1-q)$. Similarly, $\Tr[M\cN_{\varepsilon_0}(\sigma)]=y\alpha+(1-y)\beta$ with $y=sq+(1-s)(1-q)$. Since $x,y\in[1-q,q]$, we get $x\alpha+(1-x)\beta\leq q(\alpha+\beta)$ and $y\alpha+(1-y)\beta\geq (1-q)(\alpha+\beta)$. Hence
$$\Tr[M\cN_{\varepsilon_0}(\rho)]\leq \frac{q}{1-q}\Tr[M\cN_{\varepsilon_0}(\sigma)]=e^{\varepsilon_0}\Tr[M\cN_{\varepsilon_0}(\sigma)].$$
Thus $\cN_{\varepsilon_0}$ is $\varepsilon_0$-QLDP, and thus $\varepsilon^*(\cN_{\varepsilon_0})\leq \varepsilon_0$.

It remains to show that we have $\varepsilon^*(\cN_{\varepsilon_0})=\varepsilon_0$. Since $P$ is a non-trivial projection, we may choose states $\rho_0,\sigma_0\in D(A)$ such that $\Tr[P\rho_0]=1$ and $\Tr[P\sigma_0]=0$. Then $\cN_{\varepsilon_0}(\rho_0)=\tau_0$ and $\cN_{\varepsilon_0}(\sigma_0)=\tau_1$. Taking $M=|0\rangle\langle 0|$ gives $\Tr[M\cN_{\varepsilon_0}(\rho_0)]/\Tr[M\cN_{\varepsilon_0}(\sigma_0)]=q/(1-q)=e^{\varepsilon_0}$. Therefore $\cN_{\varepsilon_0}$ is not $\varepsilon$-QLDP for any $\varepsilon<\varepsilon_0$, proving $\varepsilon^*(\cN_{\varepsilon_0})=\varepsilon_0$.

Finally, consider $\varepsilon_0=\infty$ and define $$\cN_{\infty}(\rho)=\Tr[P\rho]|0\rangle\langle 0|+\Tr[(I-P)\rho]|1\rangle\langle 1|.$$ This channel is again measure-and-prepare, hence entanglement-breaking. Choose $\rho_0,\sigma_0$ as above and take $M=|0\rangle\langle 0|$. Then $\Tr[M\cN_{\infty}(\rho_0)]=1$ while $\Tr[M\cN_{\infty}(\sigma_0)]=0$. Thus no finite $\varepsilon$ can satisfy the QLDP inequality. Therefore $\varepsilon^*(\cN_{\infty})=\infty$.

\end{proof}

\section{Composition in the high privacy regime}\label{sec: 4-composition}
In this section, we prove a composition theorem for the privacy of tensor products of quantum channels that
allows arbitrary entangled pairs of joint input states and arbitrary
global output measurements, under the additional assumption that
each local channel lies in the high-privacy regime.
Throughout this section, we use the following operator characterization of QLDP proved in \cite{yoshida-hayashi,hirche2023quantum}.

\begin{lemma}[{\cite{yoshida-hayashi,hirche2023quantum}}]
\label{lem: qldp-eq}
Let $\cN\colon \cL(A)\to \cL(B)$ be a quantum channel. Then $\cN$ is $\varepsilon$-QLDP if and only if $\cN(\rho)\le e^\varepsilon \cN(\sigma)$ for every pair of states $\rho,\sigma\in D(A)$. Thus, if $\pi_A=I_A/d_A$ and $X\ge0$, then $$e^{-\varepsilon}\Tr[X]\cN(\pi_A)\le \cN(X)\le e^\varepsilon\Tr[X]\cN(\pi_A).$$
\end{lemma}

Next we prove an operator inequality for the action of a channel on the operator $R_\psi=(d+1)\psi-I_A$ for any pure state $\psi.$

\begin{lemma}\label{lem:highprivacy}
Let $\cN\colon \cL(A)\to \cL(B)$ be $\varepsilon$-QLDP, with $d=\dim A$, and assume $e^\varepsilon\le d/(d-1)$. Let $\pi=I_A/d$. For each pure state $\psi=|\psi\rangle\langle\psi|$, define $R_\psi=(d+1)\psi-I_A$ and $\tau_\psi=\cN(R_\psi)$. Then $\beta_\varepsilon \cN(\pi)\le\tau_\psi\le\gamma_\varepsilon \cN(\pi)$, where
\begin{align}
\beta_\varepsilon &\coloneqq \frac{d(d-(d-1)e^\varepsilon)}{1+(d-1)e^\varepsilon},
&
\gamma_\varepsilon &\coloneqq \frac{d(d-(d-1)e^{-\varepsilon})}{1+(d-1)e^{-\varepsilon}}.
\label{eq:beta-gamma}
\end{align}
In particular, if $e^\varepsilon<d/(d-1)$, then $\beta_\varepsilon>0$.
\end{lemma}

\begin{proof}
Fix $M\ge0$ on $B$ and set $A=\cN^\dagger(M)$. Then $A\ge0$. Let $\lambda_{\min}$ and $\lambda_{\max}$ be the minimum and maximum eigenvalues of $A$. Applying the QLDP property to the eigenvector states of $A$ gives $\lambda_{\max}\le e^\varepsilon\lambda_{\min}$.

We have $\Tr[M\tau_\psi]=\Tr[AR_\psi]=(d+1)\langle\psi|A|\psi\rangle-\Tr[A]$, while $\Tr[M\cN(\pi)]=\Tr[A]/d$.

For the upper bound, use $\langle\psi|A|\psi\rangle\le\lambda_{\max}$ and $$\Tr[A]\ge\lambda_{\max}+(d-1)\lambda_{\min}\ge\lambda_{\max}(1+(d-1)e^{-\varepsilon}).$$ Hence $\lambda_{\max}\le\Tr[A]/(1+(d-1)e^{-\varepsilon})$, and therefore $\Tr[M\tau_\psi]\le\gamma_\varepsilon\Tr[M\cN(\pi)]$.

For the lower bound, use $\langle\psi|A|\psi\rangle\ge\lambda_{\min}$ and $$\Tr[A]\le\lambda_{\min}+(d-1)\lambda_{\max}\le\lambda_{\min}(1+(d-1)e^\varepsilon).$$ Hence $\lambda_{\min}\ge\Tr[A]/(1+(d-1)e^\varepsilon)$, and therefore $\Tr[M\tau_\psi]\ge\beta_\varepsilon\Tr[M\cN(\pi)]$.

Both inequalities hold for every positive semidefinite $M$, which is equivalent to the claimed operator inequalities.
\end{proof}

\Cref{lem: qldp-eq,lem:highprivacy} allow us to prove our claimed composition theorem.

\begin{theorem}
For $i=1,\ldots,n$, let $\cN_i\colon\cL(A_i)\to \cL(B_i)$ be $\varepsilon_i$-QLDP, and write $d_i=\dim A_i$. Assume $e^{\varepsilon_i}<d_i/(d_i-1)$ for every $i$. Let $\pi_i=I_{A_i}/d_i$, and let $\beta_i,\gamma_i$ be the constants defined via \eqref{eq:beta-gamma} in \Cref{lem:highprivacy} for the channel $\cN_i$.

Then, for every joint input state $\rho_{A_1\cdots A_n}$, possibly entangled across the systems $A_1,\ldots,A_n$,
$$
\left(\prod_{i=1}^n\beta_i\right)\bigotimes_{i=1}^n\cN_i(\pi_i) \le \left(\bigotimes_{i=1}^n\cN_i\right)(\rho) \le
\left(\prod_{i=1}^n\gamma_i\right)\bigotimes_{i=1}^n\cN_i(\pi_i).
$$
Therefore, $\bigotimes_{i=1}^n\cN_i$ is $\varepsilon_{\mathrm{comp}}$-QLDP as a channel on $A_1\otimes\cdots\otimes A_n$, with respect to arbitrary pairs of joint input states, where $\varepsilon_{\mathrm{comp}}=\sum_{i=1}^n\log(\gamma_i/\beta_i)$.
\end{theorem}

\begin{proof}
For each $i$, define $R_{\psi_i}^{(i)}=(d_i+1)\psi_i-I_{A_i}$. The single-system identity from \Cref{thm: main} says that every operator $X_i\in L(A_i)$ satisfies $X_i=\int d_i\Tr[\psi_iX_i]R_{\psi_i}^{(i)}\mathbf{d}\psi_i$.

First consider two systems. Let $X\in \cL(A_1\otimes A_2)$. Apply the single-system identity only on the first tensor factor. Equivalently, apply it to each $A_1$-block of $X$. This gives
$$
X=\int d_1 R_{\psi_1}^{(1)}\otimes \Tr_{A_1}[(\psi_1\otimes I_{A_2})X]\mathbf{d}\psi_1.
$$
For each fixed $\psi_1$, the operator $\Tr_{A_1}[(\psi_1\otimes I_{A_2})X]$ lies in $\cL(A_2)$, so we may apply the single-system identity to it. Thus
$$
\Tr_{A_1}[(\psi_1\otimes I_{A_2})X] =
\int d_2\Tr[\psi_2\Tr_{A_1}[(\psi_1\otimes I_{A_2})X]]R_{\psi_2}^{(2)}\mathbf{d}\psi_2.
$$
The scalar inside the trace is $\Tr[(\psi_1\otimes\psi_2)X]$. Substituting this into the previous equation gives
$$
X=\int\int d_1d_2\Tr[(\psi_1\otimes\psi_2)X] \left(R_{\psi_1}^{(1)}\otimes R_{\psi_2}^{(2)}\right)\mathbf{d}\psi_1\mathbf{d}\psi_2.
$$
Repeating this argument one subsystem at a time gives, for every joint operator $X\in L(A_1\otimes\cdots\otimes A_n)$,
$$
X=\int\cdots\int \left(\prod_{i=1}^n d_i\right)\Tr[(\psi_1\otimes\cdots\otimes\psi_n)X] \left(R_{\psi_1}^{(1)}\otimes\cdots\otimes R_{\psi_n}^{(n)}\right) \mathbf{d}\psi_1\cdots \mathbf{d}\psi_n.
$$

Now take $X=\rho$ and define $$p_\rho(\psi_1,\ldots,\psi_n)=\left(\prod_{i=1}^n d_i\right)\Tr[(\psi_1\otimes\cdots\otimes\psi_n)\rho].$$ Since $\rho\ge0$, this function is nonnegative. Integrating over each $\psi_i$ gives $$\int\cdots\int p_\rho(\psi_1,\ldots,\psi_n)\mathbf{d}\psi_1\cdots \mathbf{d}\psi_n = \left(\prod_{i=1}^n d_i\right)\Tr[(\pi_1\otimes\cdots\otimes\pi_n)\rho]= \left(\prod_{i=1}^n d_i\right) \frac{1}{\prod_{i=1}^n d_i}=1,$$ so $p_\rho$ is a probability density.

Put $\tau_{\psi_i}^{(i)}=\cN_i(R_{\psi_i}^{(i)})$. Applying the product channel to the expansion of $\rho$ gives
$$
\left(\bigotimes_{i=1}^n\cN_i\right)(\rho) =
\int\cdots\int p_\rho(\psi_1,\ldots,\psi_n)\bigotimes_{i=1}^n\tau_{\psi_i}^{(i)}\mathbf{d}\psi_1\cdots \mathbf{d}\psi_n.
$$
By \Cref{lem:highprivacy}, $\beta_i\cN_i(\pi_i)\le\tau_{\psi_i}^{(i)}\le\gamma_i\cN_i(\pi_i)$ for every $i$ and every pure state $\psi_i$.

If $0\le X\le Y$ and $0\le X'\le Y'$, then $Y\otimes Y'-X\otimes X'=(Y-X)\otimes Y'+X\otimes(Y'-X')\ge0$, and by induction we also have $\bigotimes_i Y_i \geq \bigotimes_i X_i$ for $0 \leq X_i \leq Y_i$. Therefore,
$$
\left(\prod_{i=1}^n\beta_i\right)\bigotimes_{i=1}^n\cN_i(\pi_i) \le
\bigotimes_{i=1}^n\tau_{\psi_i}^{(i)} \le
\left(\prod_{i=1}^n\gamma_i\right)\bigotimes_{i=1}^n\cN_i(\pi_i).
$$
Integrating this operator inequality against the probability density $p_\rho$ gives the two-sided inequality.

Now let $\rho$ and $\sigma$ be arbitrary joint input states, and put $S=\bigotimes_{i=1}^n\cN_i(\pi_i)$.
The pair of inequalities gives $$\left(\bigotimes_{i=1}^n\cN_i \right)(\rho)\le\left(\prod_{i=1}^n\gamma_i\right)S$$ and $$\left(\bigotimes_{i=1}^n\cN_i\right)(\sigma)\ge \left(\prod_{i=1}^n\beta_i \right)S.$$ Since every $\beta_i$ is positive, $$\left(\bigotimes_{i=1}^n\cN_i \right)(\rho)\le \left(\prod_{i=1}^n\gamma_i/\beta_i \right) \left(\bigotimes_{i=1}^n\cN_i\right)(\sigma).$$ 
By \Cref{lem: qldp-eq}, this is exactly $\varepsilon_{\mathrm{comp}}$-QLDP.
\end{proof}

Finally, we prove an asymmetric composition theorem. That is, we determine a privacy parameter $\eps_{\text{asym}}$ such that $\cN_1 \otimes \cN_2$ is $\varepsilon_{\mathrm{asym}}$-QLDP if both $\cN_1$ and $\cN_2$ are QLDP, but only one of the channels lies in the high privacy regime.

\begin{theorem}
Let $\cN_1\colon\cL(A_1)\to \cL(B_1)$ be $\varepsilon_1$-QLDP with $e^{\varepsilon_1}<d_1/(d_1-1)$, where $d_1=\dim A_1$. Let $\beta_1,\gamma_1$ be the constants from \Cref{lem:highprivacy}. Let $\cN_2\colon\cL(A_2)\to \cL(B_2)$ be $\varepsilon_2$-QLDP for some finite $\varepsilon_2$. Then $\cN_1\otimes \cN_2$ is $\varepsilon_{\mathrm{asym}}$-QLDP on $A_1\otimes A_2$, with respect to arbitrary entangled input pairs, where $\varepsilon_{\mathrm{asym}}=\log(\gamma_1/\beta_1)+2\varepsilon_2$.
\end{theorem}

\begin{proof}
Let $\pi_i=I_{A_i}/d_i$. For a joint input state $\rho_{A_1A_2}$, use the identity from \eqref{eq:single-expansion} only on $A_1$. Define $\omega_\psi^\rho=\Tr_{A_1}[(\psi\otimes I_{A_2})\rho]$. Then $\omega_\psi^\rho\ge0$. Also $\int d_1\Tr[\omega_\psi^\rho]\mathbf{d}\psi=\Tr[\rho]=1$. The one-sided expansion is
$$
\rho=\int d_1 \left( R_\psi\otimes\omega_\psi^\rho \right) \mathbf{d}\psi.
$$
Applying $\cN_1\otimes \cN_2$ gives $(\cN_1\otimes \cN_2)(\rho)=\int d_1 \left( \tau_\psi^{(1)}\otimes \cN_2(\omega_\psi^\rho) \right) \mathbf{d}\psi$, where $\tau_\psi^{(1)}=\cN_1(R_\psi)$.

By \Cref{lem:highprivacy}, $\beta_1\cN_1(\pi_1)\le\tau_\psi^{(1)}\le\gamma_1\cN_1(\pi_1)$. By \Cref{lem: qldp-eq} applied to the positive operator $\omega_\psi^\rho$, we also have $$e^{-\varepsilon_2}\Tr[\omega_\psi^\rho]\cN_2(\pi_2)\le \cN_2(\omega_\psi^\rho)\le e^{\varepsilon_2}\Tr[\omega_\psi^\rho]\cN_2(\pi_2).$$

Combining these bounds and integrating gives $$(\cN_1\otimes \cN_2)(\rho)\le\gamma_1e^{\varepsilon_2}\cN_1(\pi_1)\otimes \cN_2(\pi_2)$$ and $$(\cN_1\otimes \cN_2)(\rho)\ge\beta_1e^{-\varepsilon_2}\cN_1(\pi_1)\otimes \cN_2(\pi_2).$$

Now let $\rho$ and $\sigma$ be arbitrary joint input states, and put $S=\cN_1(\pi_1)\otimes \cN_2(\pi_2)$. The preceding bounds give $$(\cN_1\otimes \cN_2)(\rho)\le\gamma_1e^{\varepsilon_2}S$$ and $$(\cN_1\otimes \cN_2)(\sigma)\ge\beta_1e^{-\varepsilon_2}S.$$ Hence $$(\cN_1\otimes \cN_2)(\rho)\le(\gamma_1/\beta_1)e^{2\varepsilon_2}(\cN_1\otimes \cN_2)(\sigma).$$ By \Cref{lem: qldp-eq}, this proves the claimed QLDP guarantee.
\end{proof}

\begin{remark}
    Note that this composition result has worse parameters than basic composition, but allows for entangled inputs with arbitrary global measurements. For a single $d$-dimensional
input system, let $r_d(\eps)=\log(\gamma_\eps/\beta_\eps)$, where $\beta_\eps$ and
$\gamma_\eps$ are the constants from \Cref{lem:highprivacy}. For
$0\le \eps<\log(d/(d-1))$, direct substitution gives
$$ r_d(\eps)=\log \frac{(de^\eps-d+1)(1+(d-1)e^\eps)}{(e^\eps+d-1)(d-(d-1)e^\eps)
}. $$
If all input dimensions are $d$ and all local privacy parameters are $\eps$, the composed
privacy parameter is $\eps_{\rm comp}=nr_d(\eps)$. For fixed $d$ and small $\eps$, Taylor expansion gives
$r_d(\eps)=2(d-1/d)\eps+O_d(\eps^3)$. More explicitly, $r_d(0)=0$,
$r_d'(0)=2(d-1/d)$, and $r_d''(0)=0$, so the quadratic term cancels. Thus, near
$\eps=0$, this high-privacy composition bound is linear in the number of systems and depends on $d$ as $2(d-1/d)$.

For qubit inputs, $d=2$, the expression simplifies to
$r_2(\eps)=\log((2e^\eps-1)/(2-e^\eps))$, for
$0\le \eps<\log 2$. In particular, $r_2(\eps)=3\eps+O(\eps^3)$. Hence, for $n$ identical
qubit-input channels in the strict high-privacy regime, the composition parameter is
$\eps_{\rm comp}=n\log((2e^\eps-1)/(2-e^\eps))$, which behaves as $3n\eps$ for small
$\eps$.
\end{remark}

\section{Application to private learning theory}\label{sec: 5-learning}

We now apply our results to private learning-theoretic tasks with noisy channels.
We first prove that any binary-output quantum memory protocol acting on copies of the output of an entanglement-breaking channel can be simulated exactly by a protocol that measures the corresponding unprocessed input copies one at a time and stores only classical information.

\begin{theorem}\label{lem: classical-sim-high-priv}
Let $\cQ\colon \cL(\cH_A) \to \cL(\cH_B)$ be an entanglement-breaking channel. Suppose a learner receives $T$ copies of $\cQ(\rho)$, may use arbitrary quantum memory, and finally outputs a bit. Then there is a learner without quantum memory that receives $T$ unprocessed copies of $\rho$, stores only classical outcomes, and has the same output distribution for every input state $\rho$.
\end{theorem}

\begin{proof}
Because every entanglement-breaking channel has a measure-and-prepare representation, we can write $\cQ(\rho)=\sum_{y\in Y}\Tr[E_y\rho]\sigma_y$ for some POVM $\{E_y\}_{y \in Y}$ and collection of states $\lbrace \sigma_y\rbrace_{y\in Y}$. Any binary quantum memory protocol on $T$ copies of $\cQ(\rho)$ is equivalent to a binary POVM on the $T$ output systems. Let $0\le M\le I$ be the measurement corresponding to output $1$. Its acceptance probability is
\begin{align} 
P_{\rm acc}(\rho)= \sum_{y_1,\ldots,y_T} \left(\prod_{t=1}^T\Tr[E_{y_t}\rho]\right)
\Tr\left[M(\sigma_{y_1}\otimes\cdots\otimes\sigma_{y_T})\right].
\label{eq:P-acc}
\end{align}
For a sequence $y^T=(y_1,\ldots,y_T)$, define $a(y^T)=\Tr[M(\sigma_{y_1}\otimes\cdots\otimes\sigma_{y_T})]$. Since $0\le M\le I$, we have $0\le a(y^T)\le 1$.

Now consider the following learning protocol for $\rho$ without quantum memory. On each fresh copy of $\rho$, it measures the POVM $\{E_y\}_{y\in Y}$ and stores the classical outcome. After $T$ copies, the learner has obtained a sequence $y^T$, and the protocol outputs $1$ with probability $a(y^T)$. The acceptance probability is given by the expression in \eqref{eq:P-acc}. Thus, the output distribution is the same as that for the original protocol that uses quantum memory for every $\rho$.
\end{proof}

\subsection{Purity testing with privacy}
We first recall the relevant setup from~\cite{CCHL}. In~\cite[Definition 4.15]{CCHL}, a state-learning algorithm \textit{without quantum memory} receives copies of an unknown state $\rho$ one at a time. 
On the $t$-th copy it performs an arbitrary POVM, which may depend on previous classical outcomes.
After $T$ copies, the algorithm outputs its prediction. Thus the learner stores only a classical sequence. In~\cite[Definition 4.16]{CCHL}, a state-learning algorithm \textit{with quantum memory} may store the copies of $\rho$ as quantum states. After receiving $T$ copies of the state, the learner may perform a joint POVM on $\rho^{\otimes T}$.

The purity-testing problem in Section 5.2 of~\cite{CCHL} is the following promise problem. Given copies of an unknown $n$-qubit state $\rho$, the goal is to distinguish the cases that $\rho$ is either pure or maximally mixed, $\rho=\pi_A$.

\begin{theorem}[\cite{CCHL}, Theorem 5.11]
Any learning algorithm without quantum memory that distinguishes $\rho$ pure from $\rho=\pi_A$ with success probability at least $2/3$ requires $T \ge \Omega(\sqrt d)=\Omega(2^{n/2})$ copies.
\end{theorem}

\begin{theorem}[\cite{CCHL}, Theorem 5.13]
There is a learning algorithm without quantum memory that distinguishes $\rho$ pure from $\rho=\pi_A$ using $T=O(\sqrt d)=O(2^{n/2})$ copies.
\end{theorem}

Thus, the authors of \cite{CCHL} prove that the no-quantum-memory sample complexity of this purity-testing task is $\Theta(\sqrt d)$. The same task has an $O(1)$ sample algorithm with quantum memory. This is recorded in Table 1 of~\cite{CCHL} and discussed in Section 1.1.1 of~\cite{CCHL}. Therefore the sample complexity separation for this task is $O(1)$ copies with quantum memory and $\Theta(2^{n/2})$ copies without quantum memory.
We now show that under entanglement-breaking noise, even a learner with quantum memory is subject to the noiseless single-copy measurement lower bound.

\begin{theorem}
Let $A=(\mathbb C^2)^{\otimes n}$, let $d=2^n$, and let $\pi_A=I_A/d$. For each $i=1,\ldots,n$, let $\cN_i\colon\cL(\mathbb C^2)\to \cL(B_i)$ be an entanglement-breaking channel. Put $\cQ=\cN_1\otimes\cdots\otimes\cN_n$.

Consider the purity-testing task with an additional privacy constraint in which the unknown $n$-qubit state $\rho$ is promised to be either pure or equal to $\pi_A$. The learner receives copies of $\cQ(\rho)$ and may use quantum memory. If the learner succeeds with probability at least $2/3$ for every pure state $\rho$, then $T=\Omega(\sqrt d)=\Omega(2^{n/2})$ copies are necessary.
\end{theorem}

\begin{proof}
Suppose there were such a learner using $T$ copies. By \Cref{lem: classical-sim-high-priv}, it can be simulated exactly by a no-quantum-memory learner that receives $T$ unprocessed copies of $\rho$. The simulated learner has the same success probability for every input state. In particular, it distinguishes $\rho$ pure from $\rho=\pi_A$ with success probability at least $2/3$ without quantum memory. By \cite[Theorem 5.11]{CCHL}, this requires $T=\Omega(\sqrt d)$ samples, where $d=2^n$.
\end{proof}

\subsection{Bipartite product testing with privacy}

Let $\cB_n$ denote the set of pure $n$-partite states that are
product across at least one nontrivial bipartition. Given copies of
an unknown pure state
$\psi\in\cD((\mathbb C^q)^{\otimes n})$, the bipartite
product-testing problem is to distinguish the following two cases:
\begin{enumerate}
    \item $\psi\in\cB_n$.
    \item $\psi$ is $\eta$-far away from $\cB_n$ for some fixed $\eta > 0$,    \begin{equation}\label{eq:product-testing-promise}
        \inf_{\phi\in\cB_n}
    \frac12\|\psi-\phi\|_1\geq \eta.
    \end{equation}
\end{enumerate}

Here we consider a noisy version of this task. Let $\cN\colon \cL(A)\to\cL(B)$ be a fixed known channel. The promise is imposed on the original state $\psi$, but the tester receives $T$ independent noisy copies $\cN(\psi)^{\otimes T}$. More precisely, a $T$-copy tester is a two-outcome POVM $\{M,I-M\}$ on $B^{\otimes T}$ such that
$$\operatorname{Tr}[M\cN(\psi)^{\otimes T}]\geq\frac{2}{3}$$
for every $\psi\in\cB_n$, and
$$\operatorname{Tr}[M\cN(\psi)^{\otimes T}]\leq\frac{1}{3}$$
for every pure state $\psi$ satisfying
$$\inf_{\phi\in\cB_n}\frac{1}{2}\|\psi-\phi\|_1\geq\eta.$$

\begin{theorem}
Let $A=(\mathbb C^q)^{\otimes n}$, where $q\geq 2$. For each
$i\in\{1,\ldots,n\}$, let
$\cN_i\colon\cL(\mathbb C^q)\to\cL(B_i)$
be an $\varepsilon_i$-QLDP channel satisfying
$$\varepsilon_i\leq\log\frac{q}{q-1}.$$
Set $\cQ=\cN_1\otimes\cdots\otimes\cN_n.$
Suppose that a learner receives $T$ copies of $\cQ(\psi)$,
may use arbitrary quantum memory, and solves the bipartite
product-testing problem with success probability at least $2/3$.
Then
$$T=\Omega(q^{n/4}).$$
In particular, for qubit systems, $T=\Omega(2^{n/4}).$
\end{theorem}

\begin{proof}
By the local high-privacy condition and
Theorem~\ref{thm: main}, each $\cN_i$ is
entanglement-breaking and thus also $\cQ=\bigotimes_i\cN_i$.
By \Cref{lem: classical-sim-high-priv}, the protocol with quantum memory acting on $T$ copies of $\cQ(\psi)$ has the same output distribution as a protocol that acts on one copy of $\psi$ at a time and stores only classical outcomes. This is the single-copy measurement setting for bipartite product testing, and the single-copy lower bound of
\cite[Thm~4.4]{beckey2025product} gives $T=\Omega(q^{n/4}).$
\end{proof}

\subsection{Single-copy globally private noise} 

One could consider an alternative approach to privatizing the $n$-qubit input state by applying a single privatizing channel to the entire $n$-qubit input system $A$ with $\dim A=d=2^n$. In this setting, we show that the definition of QLDP already gives a stronger lower bound in the high-privacy regime than the one obtained from the entanglement-breaking structure.

Denote by
$\Delta(\rho,\sigma):=\frac{1}{2}\|\rho-\sigma\|_1$ the trace distance between two states $\rho,\sigma$.
 Let $p\in(0,1)$, set $q=1-p$, and fix
$\alpha\in(0,pq)$. If $\cA$ is an $\varepsilon$-QLDP channel
and a binary test using $m$ copies of either $\cA(\rho)$ or
$\cA(\sigma)$ has error at most $\alpha$ with priors
$p$ and $q$, then
$$m\geq\left(1-\frac{\alpha(1-\alpha)}{pq}\right)
\frac{e^\varepsilon+1}
{2\left(e^{\varepsilon/2}-1\right)^2}
\frac{1}{\Delta(\rho,\sigma)}.$$
This is the second term in the maximum defining
$C_{\varepsilon,p,q,\alpha}$ in \cite[Theorem~4]{nuradha2024contraction}.

We apply this result as follows. Assume that the original tester uses $T$ copies with an error probability of at most $1/3$ under either hypothesis. Apply this tester independently on five sets of $T$ copies and take a majority vote. The error
probability of the resulting $5T$-copy tester is at most
$$\alpha =\sum_{j=3}^{5}\binom{5}{j}\left(\frac{1}{3}\right)^j\left(\frac{2}{3}\right)^{5-j} = \frac{17}{81} <\frac{1}{4}. $$
We may therefore apply the result from \cite{nuradha2024contraction} with equal priors $p=q=1/2$ and
$\alpha=17/81$, giving
$$ 5T\geq  \frac{2209}{13122}\frac{e^\varepsilon+1}
{\left(e^{\varepsilon/2}-1\right)^2}
\frac{1}{\Delta(\rho,\sigma)}. $$
Consequently,
$$ T\geq \frac{2209}{65610}\frac{e^\varepsilon+1}
{\left(e^{\varepsilon/2}-1\right)^2}
\frac{1}{\Delta(\rho,\sigma)}.$$
As observed in the proof of \cite[Corollary~2]{nuradha2024contraction}
$$\frac{e^\varepsilon+1}
{\left(e^{\varepsilon/2}-1\right)^2}
\geq \left( \frac{e^\varepsilon+1}{e^\varepsilon-1}
\right)^2 = \frac{1}{\tanh^2(\varepsilon/2)}. $$
Hence
\begin{equation}
     T\geq\frac{2209}{65610\,\Delta(\rho,\sigma)}\frac{1}{\tanh^2(\varepsilon/2)}. \label{eq:lower_bound_T_eps}
\end{equation}

In particular, if
$ \varepsilon\leq\log\left(\frac{d}{d-1}\right),$
then
$\frac{1}{\tanh(\varepsilon/2)} \geq 2d-1,$
and therefore
$$T\geq
\frac{2209}{65610\,\Delta(\rho,\sigma)}
(2d-1)^2.$$

We will now apply this lower bound to both tasks we have considered so far.

\begin{corollary}
\label{cor:global-private-purity}
Let $\cN\colon \cL(A)\to\cL(B)$ be $\varepsilon$-QLDP, and set $d=\dim A$. Suppose that a learner receives $T$ copies of $\cN(\rho)$ and distinguishes the case that $\rho$ is pure from the case that $\rho=\pi_A$ with success probability at least $2/3$ for every pure state $\rho$. Then
$$T\geq \frac{2209d}{65610(d-1)}\left(
\frac{e^\varepsilon+1}{e^\varepsilon-1}
\right)^2.$$
In particular, if
$\varepsilon\leq\log\!\left(\frac{d}{d-1}\right),$
then
$$T\geq\frac{2209d}{65610(d-1)}(2d-1)^2=\Omega(d^2).$$
For an $n$-qubit input, this gives $T=\Omega(4^n)$.
\end{corollary}

\begin{proof}
Fix any pure state $\psi$. A successful test must distinguish $\cN(\psi)$ from $\cN(\pi_A)$ for every pure input. Since
$\Delta(\psi,\pi_A)=1-1/d,$
the first claim follows from the discussion above leading to~\eqref{eq:lower_bound_T_eps}. Under the stated high-privacy condition (i.e., $\varepsilon\leq\log({d}/{(d-1)})$),
$$\frac{e^\varepsilon+1}{e^\varepsilon-1}\geq 2d-1,$$
which proves the second claim. 
\end{proof}

We next apply the same argument to bipartite product testing under global noise. Let $A=(\bC^q)^{\otimes n}$ and let $d=q^n$. The tester receives copies of $\cN(\psi)$, where a single $\varepsilon$-QLDP channel $\cN\colon \cL(A)\to\cL(B)$ acts on the entire $n$-partite input.

\begin{corollary}\label{cor:global-private-product}
Let $A=(\bC^q)^{\otimes n}$, where $q,n\geq 2$, and put $d=q^n$. Let $\cN\colon \cL(A)\to\cL(B)$ be an $\varepsilon$-QLDP channel. Assume that there exists a pure state $\psi\in\cD(A)$ satisfying $\inf_{\omega\in\cB_n}\Delta(\psi,\omega)\geq\eta.$
Suppose that a tester receives $T$ copies of $\cN(\xi)$, where $\xi$ is an unknown pure state, and solves the bipartite product-testing problem with success probability at least $2/3$. Then
$$T\geq\frac{2209}{65610}\frac{1}{\tanh^2(\varepsilon/2)}.$$
In particular, if
$\varepsilon\leq\log\frac{d}{d-1},$ then $T=\Omega(q^{2n}).$
\end{corollary}
\begin{proof}
Choose any pure state $\phi\in\cB_n$ and any pure state $\psi$ satisfying
$\inf_{\omega\in\cB_n}\Delta(\psi,\omega)\geq\eta.$
Since the tester succeeds with probability at least $2/3$ on every input satisfying either promise, it satisfies
$$\operatorname{Tr}[M\cN(\phi)^{\otimes T}]\geq\frac23$$
and
$$\operatorname{Tr}[M\cN(\psi)^{\otimes T}]\leq\frac13.$$
Equivalently, the tester distinguishes $\cN(\phi)^{\otimes T}$ from $\cN(\psi)^{\otimes T}$ with error probability at most $1/3$ under either hypothesis.

Applying \Cref{eq:lower_bound_T_eps} to this pair gives
$$T\geq\frac{2209}{65610\,\Delta(\phi,\psi)}\frac{1}{\tanh^2(\varepsilon/2)}.$$
Since the trace distance between two states is at most one, we get
$$T\geq\frac{2209}{65610}\frac{1}{\tanh^2(\varepsilon/2)}.$$

Now suppose that $\varepsilon\leq\log(d/(d-1))$. We again obtain
$\tanh(\varepsilon/2)\leq\frac{\frac{d}{d-1}-1}{\frac{d}{d-1}+1}=\frac{1}{2d-1}$, and so it follows that
$$T\geq\frac{2209}{65610}(2d-1)^2.$$
Finally, $d=q^n$, which gives $T=\Omega(q^{2n})$.
\end{proof}

\Cref{cor:global-private-purity} and \Cref{cor:global-private-product} explain the limitation of the preceding global channel application. In the global channel setting, the reduction using entanglement-breaking channels with \Cref{lem: classical-sim-high-priv} to the single-copy measurement setting gives the lower bound $\Omega(\sqrt d)$ for purity testing and $\Omega(q^{n/4})$ for bipartite product testing, but at the high-privacy threshold the privacy definition itself gives the stronger lower bounds $\Omega(d^2)$ and $\Omega(q^{2n})$, respectively.

\begin{table}[t]
\centering
\label{tab:private-learning-comparison}

\small
\renewcommand{\arraystretch}{0.85}
\setlength{\tabcolsep}{5pt}

\begin{tabular*}{\textwidth}{@{\extracolsep{\fill}}lccc@{}}
\toprule
Learning task & \makecell{Single-copy\\measurements}
&
Quantum memory
&
\makecell{Single-copy \\highly private channel}
\\
\midrule

Purity testing
&
\makecell{
$\Theta(2^{n/2})$\\
{\cite[Theorems~5.11 and~5.13]{CCHL}}
}
&
\makecell{
$\Theta(1)$\\
{\cite[Table~1]{CCHL}}}
&
\makecell{$\Omega(4^n)$\\
{\Cref{cor:global-private-purity}}
}
\\[4mm]

Bipartite product testing
&
\makecell{$\Omega(q^{n/4})$\\
{\cite[Theorem~4.4]{beckey2025product}}
}
&
\makecell{
Lower bound: $\Omega(n/\log n)$\\
Upper bound: $O(n/\eta^2)$\\
{\cite[Sec.~1.1]{beckey2025product}}
}
&
\makecell{$\Omega(q^{2n})$\\
{\Cref{cor:global-private-product}}
}
\\

\bottomrule
\end{tabular*}
\caption{Comparison of sample-complexity bounds for the two learning tasks considered
in this section. For bipartite product testing, $q\geq 2$ is the local dimension and $\eta>0$ is the trace-distance separation in the promise in Eq.~\eqref{eq:product-testing-promise}. The first column allows adaptive measurements of one copy
at a time with only classical information retained between copies. The second
column allows the learner to retain quantum memory and perform joint
measurements across several copies. In the last column, one
$\varepsilon$-QLDP channel acts on the entire input before each copy is
provided to a learner with arbitrary quantum memory. The lower bounds in the first column also
apply under highly private local noise.}  
\end{table}

\section{Conclusion}\label{sec: 6-conclusion}

In this work, we established the consequences of enforcing high quantum local differential privacy for preserving entanglement. We showed that if $\cN\colon \cL(A)\to\cL(B)$ is an
$\varepsilon$-QLDP channel with $d=\dim A$ and $$\varepsilon\leq\log\frac{d}{d-1},$$
then $\cN$ is entanglement-breaking.
For every value of
$\varepsilon$ strictly above it, there is a channel that is
$\varepsilon$-QLDP but \emph{not} entanglement-breaking.

We obtained two further refinements of this result. First, in the approximate
privacy setting, if $\cN$ is $(\varepsilon,\delta)$-QLDP and
$\varepsilon\leq\log(d/(d-1))$, then there is an entanglement-breaking channel
$\cM$ satisfying $$\frac12\|\cN-\cM\|_\diamond\leq(d-1)\delta.$$

Second, if one only considers entanglement with a reference
system of dimension at most $r$, then the relevant threshold for the privacy parameter is
$\varepsilon\leq\log(r/(r-1))$. In particular, the condition
$\varepsilon\leq\log 2$ destroys entanglement with every qubit reference,
regardless of the full input dimension.

We further studied composition when several private channels act on different
input systems. When every local channel lies strictly inside the
high-privacy regime, we proved a privacy guarantee for the joint
channel that applies to arbitrary pairs of joint input states. This includes
states entangled across the input systems and also allows arbitrary global
measurements on the outputs.

Finally, we applied the entanglement-breaking threshold to noisy quantum learning. Under highly private local noise, the known single-copy lower bounds continue to apply to learners with arbitrary quantum memory, giving sample-complexity lower bounds of $\Omega(2^{n/2})$ for purity testing and $\Omega(q^{n/4})$ for bipartite product testing. When one highly private channel acts on the entire input, a direct consequence of the privacy definition strengthens these lower bounds to $\Omega(4^n)$ and $\Omega(q^{2n})$, respectively.

These findings outline a boundary between classical and quantum utility in the setting of quantum local differential privacy. Because channels operating below the stated privacy threshold cannot preserve entanglement, they are unsuitable for information processing tasks that rely on resources like entanglement. Consequently, when determining privacy parameters for distributed quantum protocols, the loss of quantum resources presents a strict limit on operational utility. This also encourages the use of flexible privacy frameworks such as quantum pufferfish privacy~\cite{nuradha_QPP} instead of worst-case privacy frameworks like QLDP, whenever domain knowledge is available to harness operational utility from quantum resources.

\section*{Acknowledgments}
We thank Jacob Beckey for helpful discussions on noisy quantum learning theory. TN acknowledges support from the IQUIST Postdoctoral Fellowship from the Illinois Quantum Information Science and Technology Center at the University of Illinois Urbana-Champaign.
SB and FL were supported by National Science Foundation Grants No.~2426103 and 2442410.

\printbibliography

\appendix 
\section{Privacy near the completely depolarizing channel}\label{app: geometric-privacy}

In this appendix we relate the high-privacy threshold to separable balls around the maximally mixed state. Let $a=\dim A$, $b=\dim B$, and $D=ab$. For a channel $\cN\colon \cL(A)\to \cL(B)$, write its normalized Choi state as $J(\cN)_{AB}=(\id_A\otimes \cN)(\Phi_{AA'})$, where $|\Phi\rangle=a^{-1/2}\sum_{r=1}^a|r\rangle_A|r\rangle_{A'}$.

Let $\Delta_{A\to B}$ be the completely depolarizing channel, $\Delta_{A\to B}(X)=\Tr[X]I_B/b$. Then $J(\Delta)=I_{AB}/(ab)$. We use the following result of Gurvits and Barnum.

\begin{theorem}{\cite[Corollary 3]{gurvits2002largest}}\label{thm: gurvits-barnum}
Let $\rho$ be a state on a bipartite Hilbert space of total dimension $D$. If 
\begin{equation}\label{eq: GB-radius}
\|\rho-I_D/D\|_2\le1/\sqrt{D(D-1)},
\end{equation}
then $\rho$ is separable.
\end{theorem}

We first translate the result above to channels via the Choi isomorphism.  
\begin{proposition}
Let $\cN\colon \cL(A)\to \cL(B)$ be a quantum channel. If $$\|J(\cN)-J(\Delta)\|_2\le1/\sqrt{ab(ab-1)},$$ then $\cN$ is entanglement-breaking. In particular, the stronger condition $$\|\cN-\Delta\|_\diamond\le1/\sqrt{ab(ab-1)}$$ also implies that $\cN$ is entanglement-breaking.
\end{proposition}

\begin{proof}
The first hypothesis says exactly that the normalized Choi state $J(\cN)$ lies in the Gurvits-Barnum separable ball around $I_{AB}/(ab)$. Thus $J(\cN)$ is separable across $A:B$. Since a channel is entanglement-breaking if and only if its normalized Choi state is separable, $\cN$ is entanglement-breaking.

For the diamond-norm statement, use the inequalities $$\|J(\cN)-J(\Delta)\|_2\le\|J(\cN)-J(\Delta)\|_1\le\|\cN-\Delta\|_\diamond,$$
which prove the claim.
\end{proof}

The next result provides a bound on the optimal privacy parameter for channels that are close in diamond norm to the completely depolarizing channel.

\begin{proposition}
Suppose $\|\cN-\Delta\|_\diamond\le\eta<1/b$. Then $$\varepsilon^*(\cN)\le\log((1+b\eta)/(1-b\eta)).$$ Thus, if $$\|\cN-\Delta\|_\diamond\le r_{\mathrm{GB}}(a,b):=1/\sqrt{ab(ab-1)},$$ then $\cN$ is entanglement-breaking and $$\varepsilon^*(\cN)\le\log((1+b r_{\mathrm{GB}}(a,b))/(1-b r_{\mathrm{GB}}(a,b))).$$
\end{proposition}

\begin{proof}
For every input state $\rho$, $$\|\cN(\rho)-I_B/b\|_\infty\le\|\cN(\rho)-I_B/b\|_1\le\eta.$$ Thus $(1/b-\eta)I_B\le \cN(\rho)\le(1/b+\eta)I_B$ for every $\rho$. Hence, for all states $\rho,\sigma$, $\cN(\rho)\le((1/b+\eta)/(1/b-\eta))\cN(\sigma)$. By the operator form of QLDP proved in \Cref{lem: qldp-eq}, this is the stated privacy bound. The final statement combines this privacy bound with the previous proposition.
\end{proof}

The next result provides a bound on the optimal privacy parameter for a composition of any channel with a depolarizing channel.

\begin{proposition}
Let $\cN\colon \cL(A)\to \cL(B)$ be any quantum channel, and define $\cN_p=(1-p)\Delta+p\cN$, where $0\le p\le1$. Then $\cN_p$ is entanglement-breaking whenever $p\le1/(ab-1)$. Moreover, $$\varepsilon^*(\cN_p)\le\log(1+pb/(1-p)).$$ At the Gurvits-Barnum endpoint $p=1/(ab-1)$, assuming $ab>2$, this gives $$\varepsilon^*(\cN_p)\le\log((ab+b-2)/(ab-2)).$$
\end{proposition}

\begin{proof}
The Choi state of $\cN_p$ is $J(\cN_p)=(1-p)I_{AB}/(ab)+pJ(\cN)$. For every state $\rho$ on a $D$-dimensional Hilbert space, $$\|\rho-I_D/D\|_2\le\sqrt{1-1/D}.$$ Therefore $$\|J(\cN_p)-I_{AB}/(ab)\|_2\le p\sqrt{1-1/(ab)}.$$ The Gurvits-Barnum condition holds whenever $$p\sqrt{1-1/(ab)}\le1/\sqrt{ab(ab-1)},$$ which simplifies to $p\le1/(ab-1)$. Thus $\cN_p$ is entanglement-breaking under this condition.

For privacy, for every input state $\rho$, $\cN_p(\rho)=(1-p)I_B/b+p\cN(\rho)$. Since $0\le \cN(\rho)\le I_B$, we have $(1-p)I_B/b\le \cN_p(\rho)\le((1-p)/b+p)I_B$. Hence $$\varepsilon^*(\cN_p)\le\log(((1-p)/b+p)/((1-p)/b))=\log(1+pb/(1-p)).$$ Substituting $p=1/(ab-1)$ gives the endpoint bound.
\end{proof}

Finally, the next proposition gives a geometric perspective of the constant in \Cref{thm: main}. The optimal universal high-privacy threshold is the QLDP value at which the transpose-depolarizing Choi state parameter crosses the Gurvits-Barnum separable ball.

\begin{proposition}
Let $\cT_t(\rho)=t\rho^T+(1-t)\Tr[\rho]I/d$ be the transpose-depolarizing channel from \Cref{thm:sharp} with Choi state $J(\cT_t)$. Then $$\|J(\cT_t)-I_{d^2}/d^2\|_2=|t|\sqrt{1-1/d^2}.$$ At the parameter value $t_0=-1/(d^2-1)$ at which $\cT_t$ becomes entanglement-breaking, this distance is exactly $1/(d\sqrt{d^2-1})$, which is the Gurvits-Barnum radius for a $d\times d$ Choi state, that is, the expression on the right side of \Cref{eq: GB-radius}. Moreover, $$\varepsilon^*(\cT_{t_0})=\log(d/(d-1)).$$
\end{proposition}

\begin{proof}
Since $J(\cT_t)-I_{d^2}/d^2=t(F/d-I_{d^2}/d^2)$, and since $F^2=I_{d^2}$ and $\Tr F=d$, 
$$
\left\|\frac{F}{d}-\frac{I_{d^2}}{d^2}\right\|_2^2=\frac{\Tr F^2}{d^2} -\frac{2\Tr F}{d^3} + \frac{\Tr I_{d^2}}{d^4} = 1-\frac1{d^2}.
$$
This proves the distance formula. At $t_0=-1/(d^2-1)$, the distance becomes $$(1/(d^2-1))\sqrt{1-1/d^2}=1/(d\sqrt{d^2-1}).$$ The Gurvits-Barnum radius for total dimension $d^2$ is $$1/\sqrt{d^2(d^2-1)}=1/(d\sqrt{d^2-1}).$$

Finally, \Cref{thm:sharp} gives $\varepsilon^*(\cT_t)=\log((1-t)/(1+(d-1)t))$ on the negative branch. Substituting $t=t_0$ gives $\varepsilon^*(\cT_{t_0})=\log(d/(d-1))$.
\end{proof}

\end{document}